\documentclass[11pt]{article}

\usepackage{amsmath,amsthm,amsfonts,graphicx}
\usepackage{amssymb}
\usepackage{epsfig,graphicx,color}
\usepackage{mathrsfs}
\usepackage[active]{srcltx}
\usepackage{url}
\usepackage{enumerate}
\usepackage{authblk}
\usepackage{hyperref}
\hypersetup{
colorlinks=true,
linkcolor=blue,
filecolor=blue,
citecolor=blue,  
urlcolor=blue,
}
\usepackage{comment}

\newcommand{\R}{\mathbb{R}}

\renewcommand{\H}{\mathcal{H}}

\def\>{\rangle}
\def\<{\langle}

\newcommand{\Tr}{\operatorname{Tr}}

\newtheorem{theo}{Theorem}

\newtheorem{lemma}{Lemma}
\newtheorem{prop}{Proposition}

\newtheorem{defi}{Definition}

\newcommand{\E}{\mathbb{E}}

\newcommand{\worst}{{\operatorname{worst}}}
\newcommand{\avg}{{\operatorname{avg}}}
\newcommand{\choi}{{\operatorname{Choi}}}
\newcommand{\ol}{\overline}
\newcommand{\ul}{\underline}
\newcommand{\id}{\operatorname{id}}

\newcommand{\cE}{\mathcal{E}}
\newcommand{\cN}{\mathcal{N}}
\newcommand{\cD}{\mathcal{D}}
\newcommand{\cT}{\mathcal{T}}

\newcommand{\zw}[1]{{\color{black}#1}}
\newcommand{\zwcheck}[1]{{\color{black}#1}}
\newcommand{\linghang}[1]{{\color{black}#1}}
\newcommand{\remove}[1]{{\color{green}#1}}

\usepackage[firstpage=true]{background}

\backgroundsetup{placement=top,contents=MIT-CTP/5290,color=black,opacity=1,scale=1,position={15cm,0cm}}

\begin{document}
\title{Charge-conserving unitaries typically generate optimal covariant quantum error-correcting codes}



\newcommand{\cC}{\mathcal{C}}

\renewcommand\Affilfont{\small}

\author[1]{Linghang Kong\thanks{linghang@mit.edu}}
\author[2]{Zi-Wen Liu\thanks{zliu1@perimeterinstitute.ca}}
\affil[1]{Center for Theoretical Physics, MIT, Cambridge, MA 02139, United States}
\affil[2]{Perimeter Institute for Theoretical Physics, Waterloo, Ontario N2L 2Y5, Canada}

\date{}
\maketitle

\begin{abstract}
Quantum error correction and symmetries play central roles in quantum information science and physics.  It is known that quantum error-correcting codes covariant with respect to continuous symmetries cannot correct erasure errors perfectly (an important case being the Eastin--Knill theorem), in contrast to the case without symmetry constraints. Furthermore, there are fundamental limits on the accuracy of such covariant codes for approximate quantum error correction. Here, we consider the quantum error correction capability of random covariant codes.  In particular, we show that $U(1)$-covariant codes generated by Haar random $U(1)$-symmetric unitaries, i.e.~unitaries that commute with the charge operator (or conserve the charge), typically saturate the fundamental limits to leading order in terms of both the average- and worst-case purified distances against erasure noise.  We note that the results hold for symmetric variants of unitary 2-designs, and comment on the convergence problem of charge-conserving random circuits. Our results not only indicate (potentially efficient) randomized constructions of optimal $U(1)$-covariant codes, but also reveal fundamental properties of random charge-conserving unitaries, which may underlie important models of complex quantum systems in wide-ranging physical scenarios where conservation laws are present, such as black holes and many-body spin systems.
\end{abstract}

\section{Introduction}

\zw{
One of the most fundamental and widely studied ideas in quantum information processing is quantum error correction \cite{PhysRevA.52.R2493,nielsen2011,gottesman2010,lidar2013}, which protects (logical) quantum systems against noise and errors by suitably encoding them into quantum codes living in a larger physical Hilbert space. 
Besides the clear importance to the practical realization of quantum computing and other technologies, quantum error correction and quantum codes have also drawn great interest in physics recently as they are found to arise in many important physical scenarios in e.g.~holographic quantum gravity \cite{almheiri2015,pastawski2015} and many-body physics \cite{KITAEV20032,2015arXiv150802595Z,brandao2019}.

Physical systems typically entail certain symmetries or conservation laws, which have played central roles in the developments of many areas in physics. 
Given the broad practical and physical relevance of quantum error correction, it is important to understand its performance and limitations under symmetry constraints.  More explicitly, the encoders are restricted to be covariant with respect to some symmetry group (i.e.~commute with certain group actions), generating the so-called \emph{covariant codes} \cite{hayden2017,faist2019,Woods2020continuousgroupsof}.  
Covariant codes are known to have broad relevance in both practical and theoretical aspects, arising in many important areas in quantum information and physics such as fault tolerance \cite{eastin2009}, quantum reference frames \cite{hayden2017}, AdS/CFT correspondence \cite{harlow2018,harlow2019,kohler2019,faist2019,Woods2020continuousgroupsof}, and condensed matter physics \cite{brandao2019}.

When the symmetry group is continuous, there exists fundamental limitations on the error correction capability of the corresponding covariant codes.  The Eastin--Knill theorem \cite{eastin2009} is a well known no-go theorem in this regard, which indicates that codes covariant with respect to continuous symmetries in the sense that the logical group actions are mapped to ``transversal'' physical actions that are tensor product on physical subsystems (a feature favorable for fault tolerance for that they do not spread errors \linghang{on} one subsystem to others) cannot correct {local} errors perfectly (for physical system with finite Hilbert space dimension).  Another understanding of this phenomenon is that some logical information is necessarily leaked into the environment due to the error, which forbids perfect recovery. 
However, it is still possible to perform error correction approximately. Several ``robust'' versions of the Eastin-Knill theorem or lower bounds on the inaccuracy of covariant codes are recently found \cite{faist2019,Woods2020continuousgroupsof,KubicaD,zlj20,YangWoods}, some of which employing methods from other areas of independent interest such as quantum clocks \cite{Woods2020continuousgroupsof}, quantum metrology \cite{KubicaD,zlj20,YangWoods}, and quantum resource theory \cite{zlj20,2020arXiv201011822F}.  
  
}

\zw{
This work concerns the achievability of such lower bounds.  In particular, here we consider the simple but important case of $U(1)$ symmetry corresponding to charge (energy) conservation, which is ubiquitous in physical systems.  We analyze codes generated by unitaries drawn from the Haar measure that obey the charge conservation law, i.e.~commute with the charge operator, which are particularly interesting because: i) The results indicate typical properties of all charge-conserving unitaries due to the Haar measure; ii) Haar random unitaries and their relatives including designs and random circuits play essential roles in the study of many-body quantum systems such as black holes \cite{hayden2007,HQRY16,2019arXiv190602219H} and chaotic spin systems \cite{Nahum2,Nahum1}, indicating that our refined model with conserved quantities is potentially of broad interest in physics.
We rigorously analyze the performance of our randomly constructed covariant code against erasure noise, as characterized by both the average-case and worst-case error (measure by the purified distance) among different input states.   
To do so, we use the complementary channel technique \cite{beny2010}, which gives characterization of the error rate of a code by the amount of information leaked into the environment.  For our random code we essentially break down the error into two components, one characterizing the deviation of the random code from its average leading to error which can be bounded using a ``partial decoupling'' theorem \cite{wakakuwa2019} and turns out to be exponentially small, while the other characterizing a polynomially small intrinsic error given by the average state. 
We show that in certain important cases our random code almost always saturates the lower bounds in \cite{faist2019} to leading order, indicating that charge-conserving unitaries typically give rise to optimal covariant codes.   The results hold if the Haar random charge-conserving unitary is simplified to corresponding 2-designs, and as we conjecture, efficient random circuits composed of charge-conserving local gates.   Note that in our case with charge conservation the error is intrinsically polynomially small, while in the no-symmetry case the error of such Haar random codes is normally exponentially small and there are perfect codes.

}

\zw{
The rest of this paper is structured as follows. In Section~\ref{sec:preliminary} we formally introduce the relevant background. In Section~\ref{sec:construction} we describe the random code construction, and show that the codes are indeed covariant. In Section~\ref{sec:performance}, we present rigorous analysis of the error of our random code as measured by Choi and worst case purified distances and make a few comments about the noise model and input charge.  In Section~\ref{sec:comparison}, we explicitly compare our upper bounds on the error of our random code with known lower bounds. In Section~\ref{sec:designs}, we discuss the extension of our study to symmetric $t$-designs and random circuits as well as several associated key problems.   We conclude the work with relevant discussions in Section~\ref{sec:discussion}. Several detailed technical calculations are left to the appendices.
}

\section{Preliminaries}
Here, we formally introduce the definitions and techniques that play key roles in this work.
\label{sec:preliminary}
\subsection{Approximate quantum error correction and complementary channel formalism}

The closeness of two quantum states $\rho$ and $\sigma$ can be characterized by their fidelity
\begin{equation}\label{fidelity}
    F(\rho,\sigma) \equiv \|\sqrt{\rho}\sqrt{\sigma}\|_1 = \Tr\sqrt{\sqrt{\rho}\sigma\sqrt{\rho}}.
\end{equation}
Note that the fidelity is sometimes defined as  $F(\rho,\sigma)^2$ in the literature, but we shall stick to the definition of Eq.~(\ref{fidelity}) in this work. The \emph{purified distance} $P(\rho,\sigma)$ is then  defined as
\[
    P(\rho,\sigma) = \sqrt{1-F(\rho,\sigma)^2}.
\]
It is known \cite[Section II]{tomamichel2010} that the purified distance satisfies the triangle inequality
\[
    P(\rho,\sigma) \le P(\rho,{\tau}) + P({\tau},\sigma)
\]
for any state {$\tau$}. It satisfies the following relation with 1-norm distance \cite{nielsen2011}
\begin{equation}
    \frac{1}{2}\|\rho-\sigma\|_1  \le P(\rho,\sigma) \le \sqrt{2\|\rho-\sigma\|_1}. \label{eq:distances}
\end{equation}

\linghang{In  this work, we consider \emph{approximate} quantum error correction, whose performance is quantified by comparing the input state and the state obtained after error correction. The input state is a joint state on the logical system $L$ and a reference system $R$.} Let $\cE$ be the encoding channel of some error correction code, and $\cN$ be the noisy channel.

\zw{To characterize the overall performance of the  code, we consider two widely used strategies.
The first one makes use of the Choi isomorphism, which essentially considers a maximally entangled state as the input and characterizes the average-case error. }
More explicitly, when the reference system $R$ has the same Hilbert space dimension as $L$, we define the \linghang{Choi fidelity and} 
\zw{\emph{Choi error}}
of the code as
\[
   F_\choi \equiv \max_{\cD} F(\hat\phi^{LR}, [(\cD \circ \cN \circ \cE)^L \otimes \id^R](\hat\phi^{LR})),\quad \epsilon_\choi=\sqrt{1-F_\choi^2},
\]
where $\hat\phi^{LR}$ is the maximally entangled state between $L$ and $R$,
\[
    \hat\phi^{LR} = |\hat\phi\>\<\hat\phi|^{LR},\quad |\hat\phi\>^{LR} = \frac{1}{\sqrt{d_L}}\sum_{i=0}^{d_L-1} |i\>^L|i\>^R,
\]
and $d_L$ is the Hilbert space dimension of $L$. 
\zw{Note that this characterization of error via Choi states is closely related to the average behavior in the following sense.  It is natural to define the average-case error as 
\[
\epsilon_A:= \int d\psi P(\psi,(\cD \circ \cN \circ \cE) \psi),
\]
where $\psi^L$ denotes a pure logical state and the integral is over the uniform Haar measure.  Then we have that $\epsilon_\choi$ and $\epsilon_A$ are related by
\[
\epsilon_\choi = \sqrt{\frac{d+1}{d}}\epsilon_A,
\]
where $d$ is the dimension of the input (logical) system \cite{gilchrist2005distance,1998quant.ph..7091H}. In particular, as $d$ increases $\epsilon_\choi$ and $\epsilon_A$ approach the same value.
}
The second one is based on considering the worst-case fidelity and \zw{\emph{worst-case error}} defined as 
\[
    F_\worst \equiv  \max_{\cD}\min_{R,\rho^{LR}} F(\rho^{LR}, [(\cD \circ \cN \circ \cE)^L \otimes \id^R](\rho^{LR})),\quad \epsilon_\worst=\sqrt{1-F_\worst^2}.
\]
\linghang{Note that the minimization runs over all reference systems $R$ and all input states $\rho^{LR}$.}

The code errors $\epsilon_\choi$ and $\epsilon_\worst$ could be characterized using the formalism of complementary channels \cite{beny2010}. Let $A$ be the intermediate system that $\cN \circ \cE$ maps to. It is always possible to view $(\cN \circ \cE)^{L \to A}$ as a unitary mapping {from} $L$ to the joint system $A$ and the environment $E$, followed by a partial trace over $E$.  The complementary channel $(\widehat{\mathcal{N}\circ\mathcal{E}})^{L \to E}$ is given by tracing over $A$ and outputs the state left in the environment. Intuitively, an encoding is good if the environment does not learn much about the input. To be more precise,
\begin{align}
    \epsilon_\choi =& \min_{\zeta}P(({\widehat{\mathcal{N}\circ\mathcal{E}}}^{L\to E}\otimes \id^R)(|\hat\phi\>\<\hat\phi|^{LR}),(\mathcal{T}_\zeta^{L\to E}\otimes \id^R)(|\hat\phi\>\<\hat\phi|^{LR})), 
    \\
    \epsilon_\worst =& \min_{\zeta}\max_{\rho^{LR}} P((\widehat{\mathcal{N}\circ\mathcal{E}}^{L\to E}\otimes \id^R)(\rho^{LR}),(\mathcal{T}_\zeta^{L\to E}\otimes \id^R)(\rho^{LR})), \label{eq:complementary}
\end{align}
where $\mathcal{T}_\zeta$ is the constant channel
\[
    \mathcal{T}_\zeta^{L\to E}(\rho^L) = \Tr[\rho^L]\zeta^E.
\]
A property of the constant channel is that
\begin{equation}
    (\mathcal{T}_\zeta^{L\to E}\otimes \id^R)(\rho^{LR}) = \zeta^E \otimes \Tr_L[\rho^{LR}] \label{eq:const-channel},
\end{equation}
which will be useful for our calculations later.

\subsection{Covariant codes}
Let $G$ be a Lie group, and let $g\to U_A(g)$ and $g\to U_L(g)$ be representations of $G$ on the physical and logical Hilbert spaces respectively. A code is covariant with respect to $G$ if the encoding channel $\cE$ commutes with the representations, i.e.
\[
    \cE(U_L(g)\rho U_L(g)^\dagger) = U_A(g) \cE(\rho) U_A(g)^\dagger
\]
for all $g \in G$ and state $\rho$. \zw{A standard scenario (consider the Eastin--Knill theorem for local errors) is when  $U_A(g)$ takes the tensor product (transversal) form}
\[
    U_A(g) = U_1(g) \otimes U_2(g) \otimes \ldots \otimes U_n(g),
\]
where $U_i(g)$ acts on the $i$-th \zw{physical subsystem}.    The generators on the group $G$ will be represented as $T_L$ and $T_A$ on the logical and physical Hilbert spaces, and the tensor product form dictates that $T_A$ takes the form
\[
    T_A = \sum_{i=1}^n (T_A)_i,
\]
where $(T_A)_i$ only acts on $i$-th qubit.

\subsection{Conditional min-entropy}
\label{subsec:min-entropy}
For a bipartite state $\rho^{PQ}$, the conditional min-entropy (conditioned on $Q$) is defined as
\[
    H_{\min}(P|Q)_\rho = \sup_{\sigma\ge 0, \Tr\sigma=1}\sup\{\lambda\in \R| 2^{-\lambda}I^P\otimes \sigma^Q \ge \rho^{PQ}\}.
\]
For pure state $\psi  = |\psi\>\<\psi|^{PQ}$, there is a simple formula for $H_{\min}(P|Q)_\psi$. Let the Schmidt coefficients of $|\psi\>$ be $\alpha_1,\ldots, \alpha_D$, then the conditional min-entropy is given by
\begin{equation}
    H_{\min}(P|Q)_\psi = -2\log(\alpha_1 + \ldots + \alpha_D). \label{eq:min-max}
\end{equation}
To see this, note that \cite{konig2009} 
for any tripartite pure state $\rho$ on $X,Y$ and $Z$,
\[
    H_{\min}(X|Y)_\rho + H_{\max}(X|Z)_\rho = 0.
\]
where the conditional max-entropy $H_{\max}$ is defined as
\[
    H_{\max}(X|Z)_\rho = \sup_{\sigma^Z} \log F(\rho^{XZ},I^X \otimes \sigma^Z)^2.
\]
Note that $I^X \otimes \sigma^Z$ is not a normalized quantum state and $F(\rho^{XZ},I^X \otimes \sigma^Z)$ should be interpreted as $d_X F\left(\rho^{XZ},\frac{I^X}{d_X} \otimes \sigma^Z\right)$.

In our case the state $|\psi\>$ is pure on $P$ and $Q$, so we can choose the third register $R$ to be a trivial system, and therefore
\begin{equation}
    H_{\min}(P|Q)_\psi = -H_{\max}(P|R)_\psi = -2\log\|\sqrt{\psi^P}\sqrt{I^P}\|_1 = -2\log(\Tr[\sqrt{\psi^P}]) = -2\log(\alpha_1 + \ldots + \alpha_D). \label{eq:pure-min}
\end{equation}

\subsection{Decoupling and partial decoupling}

The (one-shot) decoupling theorem  \cite{dupuis2014} {charaterizes the degree to which a system is decoupled from the environment under certain channels in terms of (suitable variants of) conditional min-entropies.} 
It can actually be viewed as a concentration-of-measure type bound
where the randomness comes from a Haar random unitary acting on the system. To be more precise, for any bipartite state $\rho^{AR}$ and quantum channel $\tau^{A\to E}$,  
the decoupling theorem gives an upper bound for the following quantity
\[
    \E_{U^A \sim \text{Haar}}\|(\cT^{A\to E} \otimes \id^B)[U^A \rho^{AR} U^{A\dagger}]-(\cT^{A\to E} \otimes \id^B)[ \rho^{AR}_{\avg} ]\|_1
\]
where
\[
    \quad \rho^{AR}_{\avg} \equiv \E_{U^A \sim \text{Haar}} U^A \rho^{AR} U^{A\dagger}.
\]

A generalization of decoupling called \emph{partial decoupling} that will be useful for our purpose was studied in \cite{wakakuwa2019}, where the unitary $U^A$ could take a more general form. 
{We assume that} the Hilbert space of $A$ takes the form of a direct-sum-product decomposition
\[
    \H^A = \bigoplus_{j=1}^J \H_j^{A_l} \otimes \H_j^{A_r},
\]
and $U^A$ satisfies
\[
    U^A = \bigoplus_{j=1}^J  I_j^{A_l} \otimes U_j^{A_r}
\]
where $U_j^{A_r}$ is Haar random within $\H_j^{A_r}$. The distribution of such $U^A$ will be called $H_\times$. Let $l_j$ and $r_j$ be the dimensions of $\H_j^{A_l}$ and $\H_j^{A_r}$ respectively. The (non-smoothed) partial decoupling theorem \cite[Eq.~(79)]{wakakuwa2019} states that

\[
    \mathbb{E}_{U \sim H_{\mathrm{x}}}\left[\left\|\mathcal{T}^{A \rightarrow E} \circ \mathcal{U}^{A}\left(\Psi^{A R}\right)-\mathcal{T}^{A \rightarrow E}\left(\Psi_\avg^{A R}\right)\right\|_{1}\right] \leq 2^{-\frac{1}{2} H_{\min }\left(A^{*} \mid R E\right)_{\Lambda(\Psi, \mathcal{T})}}.
\]
Here, the state $\Lambda(\Psi, \mathcal{T})$ is defined as
\[
    \Lambda(\Psi, \mathcal{T}) = F(\Psi^{AR} \otimes \tau^{\bar A E})F^\dagger,
\]
where $\tau^{\bar A E}$ is the Choi-Jamiolkowski state of $\mathcal{T}$ and the operator $F^{A\bar A \to A^*}$ is
\[
    F^{A \bar{A} \rightarrow A^{*}}:=\bigoplus_{j=1}^J \sqrt{\frac{d_A l_j}{r_j}}\<\Phi_j^l|^{A_l\bar A_l}\left(\Pi_{j}^{A} \otimes \Pi_{j}^{\bar{A}}\right)
\]
with $|\Phi_j^l\>$ being the maximally mixed state. $\Pi_{j}^{A}$ is the projector into $\H_j^{A_l} \otimes \H_j^{A_r}$. Again the average state is defined as
\[
    \Psi_{\avg}^{A R} = \E_{U\sim H_\times} U \Psi^{A R} U^\dagger.
\]
and could be calculated using the formula
\begin{equation}
    \Psi_{\avg}^{A R} = \bigoplus_{j=1}^J \Psi_{jj}^{A_l R} \otimes \frac{I_j^{A_r}}{r_j}, \quad \Psi_{jj}^{A_l R} = \Tr_{A_r}[\Pi_j^A \Psi^{AR} \Pi_j^A]. \label{eq:avg-state}
\end{equation}

\section{Covariant codes from random unitaries}
\label{sec:construction}

\zw{We now define the construction of $U(1)$-covariant quantum error-correcting codes based on random charge-conserving unitaries that will be analyzed in this work.

Consider the Hamming weight operator on $m$ qubits
\[
    Q^{(m)} = \sum_{i=1}^m \frac{I - Z_i}{2},
\]
where $Z = |0\>\<0| - |1\>\<1|$ is the Pauli Z operator on a single qubit, and $Z_i$ is the operator $Z$ acting on qubit $i$. 
We consider codes that encode $k$ logical qubits in $n$ physical qubits, that are required to be covariant with respect to $U(1)$ represented by (without loss of generality)  Hamming weight operator. To be more precise, we consider Lie group $U(1)$, where the group action $e^{i\theta} \in U(1)$ is represented as $e^{i\theta} \to e^{i\theta Q^{(k)}}$ and $e^{i\theta} \to e^{i\theta Q^{(n)}}$ on the logical and physical Hilbert spaces respectively.

Our construction relies on $n$-qubit unitaries that conserve the Hamming weight, i.e.~commute with $Q^{(n)}$. Such unitaries are block diagonal with respect to the eigenspaces of $Q^{(n)}$. 
Let $H_\times$ be the Haar measure on the group of such unitaries.
We define \emph{$(n,k,\alpha)$-codes} 
as follows.

\begin{defi}
We call a code a $(n,k,\alpha)$-code if it encodes $k$ logical qubits in $n$ physical qubits by first appending an $(n-k)$-qubit state $|\psi_\alpha\>$ with Hamming weight $\alpha$, i.e.
\[
    Q^{(n-k)} |\psi_\alpha\> = \alpha |\psi_\alpha\>,
\]
and then applying a unitary $U$ on the joint $n$-qubit system that commutes with $Q^{(n)}$.  In particular, a \emph{$(n,k,\alpha)$-random code} is given by $U$ drawn from $H_\times$.
\end{defi}
It is easy to verify that such codes indeed satisfy the covariance condition.
\begin{prop}
\label{thm:cov}
 $(n,k,\alpha)$-codes are covariant with respect to a $U(1)$ symmetry, with the logical charge operator $T_L=Q^{(k)}$ and physical charge operator $T_A = Q^{(n)}$. 
Note that this property holds for the $(n,k,\alpha)$-random code.
\end{prop}
}
\begin{proof}
Since the $n$-qubit unitary $U$ commutes with $Q^{(n)}$, it commutes with $e^{iQ^{(n)}\theta}$ for all $\theta$. Then for any $k$-qubit logical state $\rho$ we have
\begin{align*}
    e^{i Q^{(n)} \theta} U(\rho \otimes |\psi_\alpha\>\<\psi_\alpha|)U^\dagger e^{-iQ^{(n)}\theta} =& U e^{i Q^{(n)} \theta} (\rho \otimes |\psi_\alpha\>\<\psi_\alpha|) e^{-iQ^{(n)}\theta} U^\dagger \\=& U  (e^{i Q^{(k)} \theta}\rho e^{-i Q^{(k)} \theta} \otimes e^{i Q^{(n-k)} \theta}|\psi_\alpha\>\<\psi_\alpha|e^{-i Q^{(n-k)} \theta})  U^\dagger\\
    =& U  (e^{i Q^{(k)} \theta}\rho e^{-i Q^{(k)} \theta} \otimes |\psi_\alpha\>\<\psi_\alpha|)  U^\dagger,
\end{align*}
which means the encoding map $\cE$ satisfies the covariance condition
\[
    e^{i Q^{(n)} \theta} \cE(\rho) e^{-i Q^{(n)} \theta} = \cE(e^{i Q^{(k)} \theta}\rho e^{-i Q^{(k)} \theta}).
\]
\end{proof}
\zw{Again, note that the $(n,k,\alpha)$-random code can be viewed as a randomized construction of covariant code, and in addition, it reveals the average, or typical, properties of charge-conserving unitaries due to the Haar measure.
}

\section{Performance of random covariant codes}\label{sec:performance}
We now derive explicit bounds on both Choi and worst-case errors of $(n,k,\alpha)$-random code against erasure of $t$ qubits in terms of purified distance.  
We will fix the set of erased qubits in our analysis, but the results still hold when the $t$ qubits are picked randomly, as will be discussed in Section~\ref{sec:remarks}.
\subsection{Choi error}
\label{sec:choi}

In this part we analyze the Choi error of $(n,k,\alpha)$-codes. Since the noise channel we consider is an erasure of $t$ qubits, the complementary channel will be a partial trace over the $n-t$ unaffected qubits. We denote this by $\Tr_{n-t}[\cdot]$. We obtain the following result.

\begin{theo}
\label{thm:choi}
\linghang{
In the large $n$ limit, when $k$ and $t$ satisfy $k^2t^2 = o(n)$ and $\alpha=a n$ with $0<a<1$ being a constant, the \zw{expected Choi error of the $(n,k,\alpha)$-random code} satisfies}
\begin{equation}
    \E \epsilon_\choi \le  \frac{\sqrt{tk}}{4n\sqrt{a(1-a)}}\left(1+O\left(\frac{k^2t^2}{n}\right)\right).
\end{equation}
Furthermore, the probability that \zw{the Choi error of a $(n,k,\alpha)$-code (with respect to $H_\times$)}  violates the  inequality above is exponentially small, i.e.
\begin{equation}
\Pr\left[\epsilon_\choi >  \frac{\sqrt{tk}}{4n\sqrt{a(1-a)}}\left(1+O\left(\frac{k^2t^2}{n}\right)\right)\right] = e^{-\Omega(n)}.
\end{equation}
\end{theo}

\begin{proof}
By Eq.~\eqref{eq:complementary}, for a specific choice of $U$ in our $(n,k,\alpha)$-code construction, the Choi error of the corresponding code is given by
\begin{align}
    \epsilon_\choi =& \min_\zeta P\left(\Tr_{n-t}[U(|\hat\phi\>\<\hat\phi|\otimes |\psi_\alpha\>\<\psi_\alpha|)U^\dagger],\frac{I}{2^k}\otimes \zeta\right) \nonumber\\
    \le& P\left(\Tr_{n-t}[U(|\hat\phi\>\<\hat\phi|\otimes |\psi_\alpha\>\<\psi_\alpha|)U^\dagger],\Tr_{n-t}\Phi_\text{avg}\right)+\min_\zeta P\left(\Tr_{n-t}\Phi_\text{avg},\frac{I}{2^k}\otimes \zeta\right) \nonumber \\
    \le& \sqrt{2\left\|\Tr_{n-t}[U(|\hat\phi\>\<\hat\phi|\otimes |\psi_\alpha\>\<\psi_\alpha|)U^\dagger]-\Tr_{n-t}\Phi_\text{avg}\right\|_1}+\min_\zeta P\left(\Tr_{n-t}\Phi_\text{avg},\frac{I}{2^k}\otimes \zeta\right), \label{eq:bound}
\end{align}
where the average state is
\[
    \Phi_\text{avg} = \E_{U\sim H_\times} U(|\hat\phi\>\<\hat\phi|\otimes |\psi_\alpha\>\<\psi_\alpha|)U^\dagger.
\]
Note that Eq.~\eqref{eq:const-channel} has been used in the first line, \zw{and the second line follows from the triangle inequality of $P$}.

When averaging over $U$ sampled from $H_\times$, the first term in Eq.~\eqref{eq:bound} can be bounded using the partial decoupling theorem:
\begin{align*}
    &\E_{U\sim H_\times} \sqrt{2\left\|\Tr_{n-t}[U(|\hat\phi\>\<\hat\phi|\otimes |\psi_\alpha\>\<\psi_\alpha|)U^\dagger]-\Tr_{n-t}\Phi_\text{avg}\right\|_1} \\
    \le& \sqrt{2\E_{U\sim H_\times}\left\|\Tr_{n-t}[U(|\hat\phi\>\<\hat\phi|\otimes |\psi_\alpha\>\<\psi_\alpha|)U^\dagger]-\Tr_{n-t}\Phi_\text{avg}\right\|_1} \\
    \le & \sqrt{2} \times 2^{-\frac{1}{4}H_{\min}(A^*|RE)_\Lambda}.
\end{align*}

We prove in Appendix~\ref{app:entropy} that
\[
    H_{\min }(A^*|RE)_\Lambda = \Omega(n),
\]
which implies that the expectation value of the first term is exponentially small. A simple application of the Markov's inequality shows that this term is exponentially small with probability equal to one minus an exponentially small amount.

The second term in Eq.~\eqref{eq:bound} is independent of $U$. Since this is a minimization, an upper bound on this term can be found by any choice of $\zeta$. We set $\zeta$ to be the $t$-qubit marginal state of $\Tr_{n-t}[\Phi_\text{avg}]$, and as detailed in Appendix~\ref{app:average-state} we obtain
\[
    P\left(\Tr_{n-t}\Phi_\text{avg},\frac{I}{2^k}\otimes \zeta\right) \le \frac{\sqrt{tk}}{4n\sqrt{a(1-a)}}\left(1+O\left(\frac{t^2k^2}{n}\right)\right).
\]


\end{proof}

\subsection{Worst-case error}
\label{sec:worst-case}
Here we consider the worst-case error of $(n,k,\alpha)$-codes under erasure of $t$ qubits. The following lemma gives an lower bound of worst case purified distance for a fixed code.
\begin{lemma}[{\cite[Thm.~3]{faist2019}}]
\label{lem:worst-bound}
For any encoding channel $\mathcal{E}$ and  any noise channel $\mathcal{N}$, let $\widehat{\mathcal{N}\circ\mathcal{E}}$ be a complementary channel of $\mathcal{N}\circ\mathcal{E}$. Fixing a basis of logical states $\{|x\>\}$, we define
\begin{equation}
    \rho^{x,x'} = \widehat{\mathcal{N}\circ\mathcal{E}} (|x\>\<x'|). \label{eq:rho-def}
\end{equation}
Assume that there exists a state $\zeta$, as well as constants $\epsilon,\epsilon'>0$ such that
\begin{align}
    P(\rho^{x,x}, \zeta) \le& \epsilon, \label{eq:rhoxx}\\
    \|\rho^{x,x'}\|_1 \le & \epsilon', \quad \forall x \not= x'. \label{eq:rhoxxp}
\end{align}
Then, the code $\mathcal{E}$ is an approximate error-correcting code with an approximation parameter satisfying
\begin{equation}
    \epsilon_\worst \le \epsilon+ d_L \sqrt{\epsilon'}, \label{eq:worst-bound}
\end{equation}
where $d_L$ is the dimension of the logical system.

If one of several noise channels is applied at random but it is known which one occurred, then Eq.~\eqref{eq:worst-bound} holds for the overall noise channel if the assumptions above are satisfied for each individual noise channel.
\end{lemma}
For our $(n,k,\alpha)$-code construction, 
\[
 {\rho^{x,x'}} = \Tr_{n-t}[U(|x\>\<x'| \otimes |\psi_\alpha\>\<\psi_\alpha|)U^\dagger].
\]
Note that the lemma above applies to a fixed encoding $\mathcal{E}$. To generalize this theorem to our randomized construction, we define $\rho^{x,x'}_\avg$ as $\rho^{x,x'}$ in Eq.~\eqref{eq:rho-def} averaged over the random unitary in $\mathcal{E}$. Then using the following lemma, we can obtain bounds on the worst-case error. 
\begin{lemma}
Consider the large $n$ limit where parameters $\alpha$, $k$, $t$, $\epsilon$, $\delta$ and $\delta'$ each might depend on $n$. If the average states $\rho^{x,x'}_\avg$ satisfy $P(\rho^{x,x}_\avg, \zeta) \le \epsilon$ for some fixed state $\zeta$ independent of $x$, then with probability at least $1-p_1-p_2$ the code generated by our construction satisfy
\[
    \epsilon_\worst \le \epsilon + \delta + 2^k \sqrt{\delta'}.
\]
Here, $p_1$ and $p_2$ are defined as
\begin{align}
    \log p_1 =& k-\frac{n}{4}\min\left\{H\left(\frac{\alpha}{n}\right),H\left(\frac{\alpha+k}{n}\right)\right\} + \frac{t}{2}+\log\frac{1}{\delta} + O(\log n),\nonumber\\
    \log p_2 =& 2k-\frac{n}{2} \min\left\{H\left(\frac{\alpha}{n}\right),H\left(\frac{\alpha+k}{n}\right)\right\} + t + \log\frac{1}{\delta'} + O(\log n). \label{eq:p1p2}
\end{align}
\label{lem:worst-randomized}
\end{lemma}
\begin{proof}
We use the partial decoupling theorem to find an upper bound for the average distance between $\rho^{x,x'}_\avg$ and $\rho^{x,x'}$. Then by Markov inequality this bounds the probability that $\rho^{x,x'}$ behaves much worse than $\rho^{x,x'}_\avg$. Finally we can show the code has good performance with high probability using a union bound.

Given that there exists $\zeta$ such that $P(\rho^{x,x}_\avg, \zeta) \le \epsilon$ for all $x$, we have
\begin{align*}
    \E_U P(\rho^{x,x}, \zeta) \le& P(\rho^{x,x}_\avg, \zeta) + \E_U P(\rho^{x,x}, \rho^{x,x}_\avg) \\
    \le& \epsilon + \E_U \sqrt{2\|\rho^{x,x}- \rho^{x,x}_\avg\|_1} \\
    \le& \epsilon +  \sqrt{2\E_U\|\rho^{x,x}- \rho^{x,x}_\avg\|_1} \\
    \le& \epsilon + \sqrt{2} \times 2^{-\frac{1}{4}H_{\min}^{x}}.
\end{align*}
where $H_{\min}^{x}$ is $H_{\min }\left(A^{*} \mid R E\right)_{\Lambda(\Psi, \mathcal{T})}$ with the initial state set to be $|x\>\<x|$. By Appendix~\ref{app:worst-entropy},
\[
    H_{\min}^x = nH\left(\frac{|x|+\alpha}{n}\right)-2t+\log(n)\ge n \min\left\{H\left(\frac{\alpha}{n}\right),H\left(\frac{\alpha+k}{n}\right)\right\} - 2t + O(\log n),
\]
which means that for each $x$,
\begin{align}
    & \log \Pr[P(\rho^{x,x},\zeta) \ge \epsilon+\delta] \nonumber \\
    \le & \, \frac{1}{2}-\frac{1}{4}H_{\min}^x + \log\frac{1}{\delta} \nonumber \\
    \le & -\frac{n}{4}\min\left\{H\left(\frac{\alpha}{n}\right),H\left(\frac{\alpha+k}{n}\right)\right\} + \frac{t}{2}+\log\frac{1}{\delta} + O(\log n). \label{eq:prob-x}
\end{align}

When $x\not=x'$, it is easy to see that $\rho^{x,x'}_\avg=\Tr_{n-t}[U(|x\>\<x'| \otimes |\psi_\alpha\>\<\psi_\alpha|)U^\dagger]=0$. Since the partial decoupling theorem only applies to subnormalized states, we need to write $|x\>\<x'|$ using the following relation
\[
    |x\>\<x'| = \frac{1}{2}|\mu_{x,x'}^+\>\<\mu_{x,x'}^+| - \frac{1}{2}|\mu_{x,x'}^-\>\<\mu_{x,x'}^-| + \frac{i}{2}|\nu_{x,x'}^+\>\<\nu_{x,x'}^+| - \frac{i}{2}|\nu_{x,x'}^-\>\<\nu_{x,x'}^-|,
\]
where
\[
    |\mu^\pm_{x,x'}\> = \frac{1}{\sqrt 2}(|x\> \pm |x'\>), \quad|\nu^\pm_{x,x'}\> = \frac{1}{\sqrt 2}(|x\> \pm i|x'\>).
\]
Then we can apply the partial decoupling theorem to the states $|\mu^\pm_{x,x'}\>$ and $|\nu^\pm_{x,x'}\>$, and have
\[
    \E_U \|\rho^{x,x'}\|_1 \le \frac{1}{2}\left(2^{-\frac{1}{2}H_{\min}^{x,x',\mu^+}}+2^{-\frac{1}{2}H_{\min}^{x,x',\mu^-}}+2^{-\frac{1}{2}H_{\min}^{x,x',\nu^+}}+2^{-\frac{1}{2}H_{\min}^{x,x',\nu^-}}\right),
\]
where $H_{\min}^{x,x',\mu^\pm}$ and $H_{\min}^{x,x',\nu^\pm}$ are $H_{\min }\left(A^{*} \mid R E\right)_{\Lambda(\Psi, \mathcal{T})}$ with the initial state set to be $|\mu^\pm_{x,x'}\>$ and $|\nu^\pm_{x,x'}\>$ respectively.
From Appendix~\ref{app:worst-entropy} we know that
\[
    H_{\min}^{x,x'} \ge n \min\left\{H\left(\frac{\alpha}{n}\right),H\left(\frac{\alpha+k}{n}\right)\right\} - 2t + O(\log n),
\]
where $H_{\min}^{x,x'}$ could be any one among $H_{\min}^{x,x',\mu^\pm}$ and $H_{\min}^{x,x',\nu^\pm}$. This gives the bound
\begin{equation}
    \log \Pr[\|\rho^{x,x'}\|_1 \ge \delta'] \le -\frac{n}{2} \min\left\{H\left(\frac{\alpha}{n}\right),H\left(\frac{\alpha+k}{n}\right)\right\} + t + \log\frac{1}{\delta'} + O(\log n). \label{eq:prob-xx}
\end{equation}
Now we can apply union bound and take sum of Eq.~\eqref{eq:prob-x} over all $x$ and Eq.~\eqref{eq:prob-xx} over all $x$ and $x'$. This means that $P(\rho^{x,x},\zeta) \le \epsilon+\delta$ and $\|\rho^{x,x'}\|_1 \le \delta'$ are satisfied for all $x$ and $x'$ with probability at least $1-p_1-p_2$ with $p_1$ and $p_2$ defined in Eq.~\eqref{eq:p1p2}. By Lemma~\ref{lem:worst-bound}, the code satisfies
\[
    \epsilon_\worst \le \epsilon + \delta + 2^k \sqrt{\delta'}.
\]
\end{proof}
\begin{theo}
\label{thm:worst}
\linghang{In the large $n$ limit}, 
when $k$ and $t$ satisfy $kt^2=o(n)$ and $\alpha=an$ with $0<a<1$ being a constant, \zw{the expected worst-case error of the $(n,k,\alpha)$-random code} satisfies \begin{equation}
    \E \epsilon_\worst \le \frac{k\sqrt{t}}{4n\sqrt{a(1-a)}}\left(1+O\left(\frac{kt^2}{n}\right)\right). \label{eq:worst-bound2}
\end{equation}
Furthermore, the probability 
\zw{\zw{that the worst-case error of a $(n,k,\alpha)$-code (with respect to $H_\times$)}  violates the  inequality above is exponentially small, i.e.}
\begin{equation}
    \Pr\left[\epsilon_\worst > \frac{k\sqrt{t}}{4n\sqrt{a(1-a)}}\left(1+O\left(\frac{kt^2}{n}\right)\right)\right]= e^{-\Omega(n)}.
\end{equation}
\end{theo}
\begin{proof}
In Appendix~\ref{app:worst-avg}, we can show that $P(\rho^{x,x}_\avg, \zeta)$ is upper bounded by the right hand side of Eq.~\eqref{eq:worst-bound2}. Now we apply Lemma~\ref{lem:worst-randomized} with properly chosen exponentially small $\delta$ and $\delta'$,so that $p_1$ and $p_2$ are also exponentially small, which shows that the code satisfies $\epsilon_\worst = O(n^{-1})$ \zw{with exponentially small failure probability}. \linghang{Since $\epsilon_\worst$ is at most 1, this exponentially probability of violation implies that the expectation of $\epsilon_\worst$ satisfies the inequality as well.}
\end{proof}

\subsection{Remarks on noise and charge}
\label{sec:remarks}

We have shown in Thm.~\ref{thm:choi} and Thm.~\ref{thm:worst} that when $\alpha=an$, $0<a<1$, the $(n,k,\alpha)$-random code satisfies
\[
    \epsilon_\choi \le  \frac{\sqrt{tk}}{4n\sqrt{a(1-a)}}\left(1+O\left(\frac{k^2t^2}{n}\right)\right),\quad
    \epsilon_\worst \le \frac{k\sqrt{t}}{4n\sqrt{a(1-a)}}\left(1+O\left(\frac{kt^2}{n}\right)\right),
\]
with high probability against erasure of $t$ qubits, where the probability of failure (greater error) is exponentially small in $n$.
Note that although the above analyses are done with respect to the erasure of $t$ specific qubits, the bounds hold for the more general case of the erasure of any combination of $t$ qubits with some probability. To show this, first note that the purified distance is independent of the choice of $t$ qubits, because the permutation of the qubits commutes with the Hamming weight operator and thus could be absorbed into the Haar measure. Then a union bound could be used over all choices of $t$ qubits, which \zw{at most amplifies} the failure probability by a factor of $\binom{n}{t}$. Since $t=o(n)$ is obviously needed for the bounds to give meaningful result, the failure probability is still exponentially small as $\log\binom{n}{t}=o(n)$.

Instead of erasing $t$ out of the $n$ qubits, a natural and stronger model of erasure noise is to have each qubit erased with some independent probability. Our analysis can be easily applied to this model by combining our bounds with the distribution on the number of qubits erased.

It is also interesting to note that the code performance depends on the charge of the input ancilla state in our construction.  In particular, $a=1/2$ (namely $\alpha=n/2$) gives rise to the best accuracy, and the accuracy becomes worse as one increases or decreases $a$. The intuition is that  the code resides in a subspace with Hamming weight between $\alpha$ and $\alpha+k$. When $\alpha=n/2$ (note that $k=o(n)$), the full Hilbert space with Hamming weight between $\alpha$ and $\alpha+k$ is the largest and apparently the most entangled, which enhances the performance of the code.

\zw{
How does our approach apply to the case without charge conservation laws? 
There, our construction is modified by replacing the charge-conserving Haar-random unitary by a fully Haar random one in $SU(2^n)$. As long as the quantum singleton bound $n-k \ge 4t$ is satisfied, the code obtained normally has an error rate exponentially small in $n$, in contrast to polynomial small for the case with charge conservation. 
To be more explicit, let $\Delta = n-k-4t$, then the randomly generated code will have expected Choi error $e^{-\Omega(\Delta)}$. This could be shown using an analysis similar to the case with symmetry constraint, and the main difference is that the term relevant to the average state (i.e.\ the second term in Eq.~\eqref{eq:bound}) is naturally zero in this case and the error solely comes from the decoupling bound. 
}
{This comparison also gives an intuition for the Eastin-Knill theorem and the lower bounds for covariant codes from a mathematical perspective.}

\zw{
Another interesting thing to note is the dependence of error scaling on charge scaling as well as the choice of distance measure.  Here we discuss the Choi case.
As we have shown, the first term of Eq.~(\ref{eq:bound}) is exponentially small in $n$ for linear $\alpha$. When $\alpha$ is constant, following the calculation in Appendix~\ref{app:entropy} the relative entropy is at least $\alpha\log n$, so the second term is dominant and determines the overall distance. Now the second term with $\zeta$ set as $\Tr_R\Tr_{n-t}\Phi_\avg$ behaves significantly differently in 1-norm and purified distance for constant charge. According to our numerical results as illustrated in Fig.~\ref{fig:distance}, the purified distance scales worse for constant $\alpha$ than for linear $\alpha$. To be more precise, by a linear fitting shown in Table~\ref{tab:slope} it can be seen that the error measured by purified distance scales roughly like $n^{-1/2}$ when $\alpha = O(1)$ and $n^{-1}$ when $\alpha= O(n)$. On the other hand, the error measured by trace distance always scales like $n^{-1}$. The two cases match the two extremes in Eq.~\eqref{eq:distances}.} 
\linghang{These numerical results are consistent with our calculation in Appendix~\ref{app:average-state}.}

\begin{figure}
    \centering
    \includegraphics[width=0.49\textwidth]{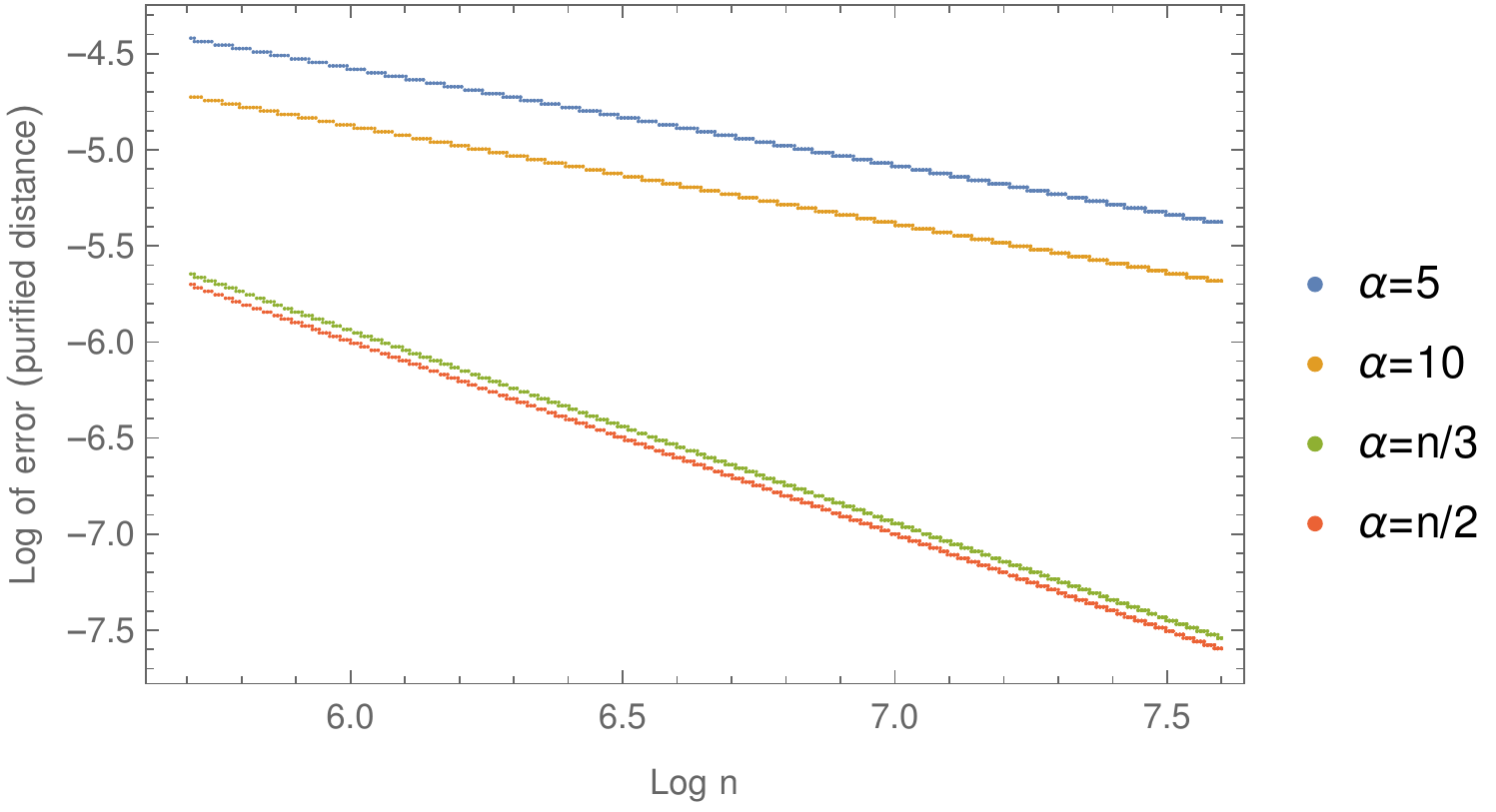}
    \includegraphics[width=0.49\textwidth]{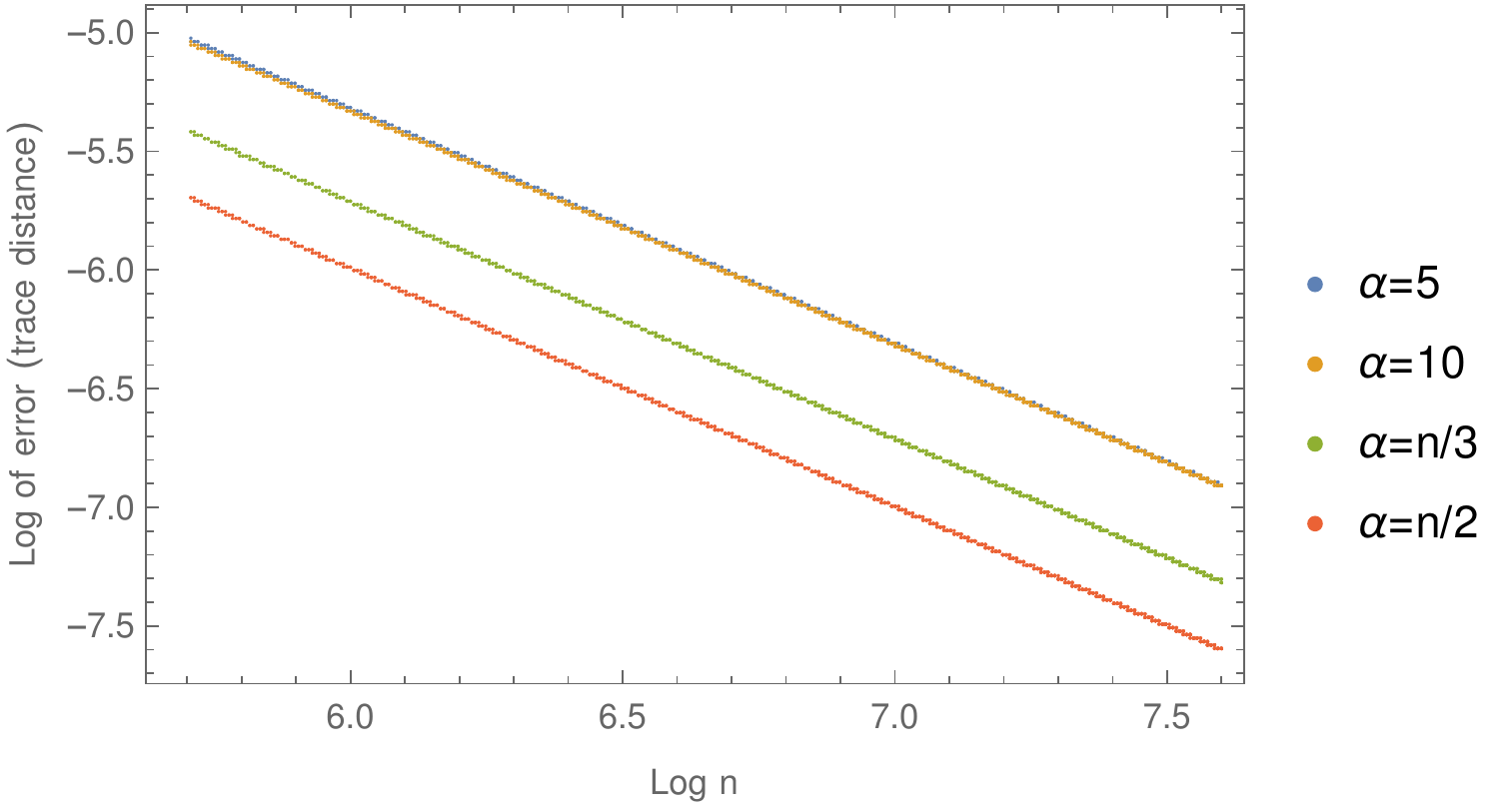}
    \caption{Log-log plot of \zw{the Choi error of $(n,k,\alpha)$-random code as measured by 1-norm distance and purified distance, given different $\alpha$; here we set $k=t=2$.} 
    }
    \label{fig:distance}
\end{figure}
\begin{table}
\centering
\begin{tabular}{c|c|c|c|c}
    \hline
    $\alpha$ & 1 & 5 & $n/3$ & $n/2$\\
    \hline
    1-norm distance & -0.998743 & -0.993676 & -0.999996 & -1.00313 \\
    \hline
    purified distance & -0.503147 & -0.504438 & -1.00001 & -1.00063 \\
    \hline
\end{tabular}
\caption{The slopes for the lines in Fig.~\ref{fig:distance}.}
\label{tab:slope}
\end{table}

\section{Comparisons with fundamental limits}
\label{sec:comparison}

Now let us compare the performance of our $(n,k,\alpha)$-codes with the fundamental limits of covariant codes.
For simplicity, consider the $t=1$ case, namely when one qubit is erased.
For $U(1)$ symmetry, the Thm.~1 of \cite{faist2019} indicate the following lower bounds:
\begin{equation}
    \epsilon_\choi \ge \frac{\binom{k}{\lceil k/2 \rceil}\lceil k/2 \rceil}{2^k n},\quad \epsilon_\worst \ge \frac{k}{2n}. \label{eq:lower-bounds}
\end{equation}
If $k$ is large, the bound on $\epsilon_\choi$ {approaches}
\[
    \epsilon_\choi \ge \frac{1}{n}\sqrt{\frac{k}{2\pi}}.
\]
In comparison, from Thm.~\ref{thm:choi} and Thm.~\ref{thm:worst}, our $(n,k,\alpha)$-random code has smallest error when $a=1/2$, and in that case
\[
    \epsilon_\choi \le \frac{\sqrt{k}}{2n}(1+O(k^2/n)), \quad \epsilon_\worst \le \frac{k}{2n}(1+O(k/n)),
\]
so {up to leading order}, the worst-case distance exactly matches the lower bound, and the Choi distance matches the bound up to a constant. When $k=1$, the Choi purified distance also matches the bound in Eq.~\eqref{eq:lower-bounds} exactly.

The situation for the general $t>1$  case is as follows.
Thm.~2  in Ref.~\cite{faist2019} gave a bound for general erasure: If the physical charge operator has the form
\[
    T_A = \sum_\alpha T_\alpha,
\]
where $T_\alpha$ has support on a set of qubits $\alpha$, and that $\alpha$ gets erased with probability $q_\alpha$, then
\begin{equation}\label{eq:t_bound}
    \epsilon_\choi \ge \frac{\left\|T_{L}-\mu(T_{L}) I_{L}\right\|_{1} / d_{L}}{\max _{\alpha}\left(\Delta T_{\alpha} / q_{\alpha}\right)},\quad \epsilon_\worst \ge \frac{\Delta T_{L} }{2\max _{\alpha}\left(\Delta T_{\alpha} / q_{\alpha}\right)},
\end{equation}
where $\Delta T_\alpha$ and $\Delta T_L$ are the difference between the largest and smallest eigenvalues of $T_\alpha$ and $T_L$ respectively, and $\mu(T_L)$ is the median of the eigenvalues of $T_L$.
For example, consider the following two ways to model the erasure of $t$ qubits:
\begin{enumerate}
    \item The qubits are grouped into sets of size $t$, and there are $n/t$ such sets. Each $T_\alpha$ is the Hamming weight operator on this set, and $\Delta T_\alpha = t, q_\alpha = t/n$;
    \item $\alpha$ represent all possible sets of $t$ qubits. $q_\alpha=\binom{n}{t}^{-1}$, and each $T_\alpha$ is $\frac{n}{t}\times \binom{n}{t}^{-1}$ the Hamming weight operator on the set, where the coefficient is chosen so that $T_\alpha$ sum up to the physical charge operate $T_A$. In this case $\Delta T_\alpha = n\times \binom{n}{t}^{-1}, q_\alpha=\binom{n}{t}^{-1}$.
\end{enumerate}
In either case, it turns out that
\[
    \max_\alpha \left(\Delta T_{\alpha} / q_{\alpha}\right) = n.
\]
As a result, the lower bounds Eq.~(\ref{eq:t_bound}) do not scale with $t$ and are expected to be loose. 
\zw{We leave potential improvement of the lower bounds in \cite{faist2019} as well as more careful investigation into different methods in e.g.~\cite{Woods2020continuousgroupsof,YangWoods,KubicaD,zlj20} for the $t>1$ case for future work.}

\section{On designs and random circuits with conservation laws}
\label{sec:designs}
\zw{In the above, we considered Haar random unitaries with charge conservation, for which one important motivation is to understand the typical performance of all such unitaries.     A natural follow-up question of both practical and mathematical interest is how well the  results  hold for more ``efficient'' versions of random unitaries such as $t$-designs (``pseudorandom'' distributions that match the Haar measure up to $t$ moments) and random quantum circuits (circuits composed of random local gates).  Consider the case without symmetries---it is known that the decoupling and error correction properties of Haar random unitaries hold for (approximate) 2-designs \cite{dupuis2014}, and that random circuits converge to $t$-designs in depth polynomial in $t$ and $n$ \cite{harrow2009,brandao2016}, which imply that random circuits can provide rather efficient constructions of good codes (see also \cite{BrownFawzi:decoupling,Brown13shortrandom}).    Do similar conclusions hold for the case with charge conservation?

}
In our analysis, the Haar randomness has only been used in the partial decoupling theorem, and as was noted in \cite{wakakuwa2019}, 
\zw{symmetric 2-designs 
are sufficient for  the partial decoupling bounds to hold. Therefore, all our bounds hold for 2-designs with charge conservation. 
}
Although 2-designs 
for the full unitary group $SU(2^n)$ have been widely studied \cite{harrow2009,brandao2016,cleve2015,nakata2017}, little is known about 2-designs (let alone higher order designs) with symmetry constraints, as is in our case. 
\zw{Especially for the fundamental problem of convergence of random circuits to designs, we note a few interesting differences.}
Repeated applications of 2-qubit charge-conserving random unitary gates may converge to 2-designs that we need, as was in \cite{harrow2009,brandao2016} for the no-symmetry case, but the proof techniques there might be difficult to be adapted to this problem -- negative values could appear in the operator basis, making the Markov chain analysis difficult; The conversion into Hamiltonian spectrum problem in \cite{brandao2016} does not work here either, due to the complicated structure of the eigenspaces of Hamming weight operator. \zw{To summarize, there seem to be nontrivial obstacles to adapting previous proofs of convergence of random circuits to the case with charge conservation, and it remains open how to efficiently construct 2-designs with charge  conservation.  }
\zw{On the other hand, it is recently shown that \cite{Marvian} in the presence of continuous symmetries like $U(1)$, the group of unitaries generated by local symmetric gates is a proper subgroup of the group of  global symmetric ones.  As a result, under charge conservation, local random circuits cannot converge to the Haar measure, and it remains to be further studied whether they converge to certain $t$-designs.    For the weaker error correction property, we conjecture that
random circuits composed of charge-conserving local gates are able to approach the near-optimal performance of Haar random unitaries derived in Thm.~\ref{thm:choi} and Thm.~\ref{thm:worst} with an efficiency similar to the no-symmetry case, that is, $\widetilde{O}(n)$ gates or $O(\mathrm{polylog}(n)$) depth (for circuit architecture without geometries). 
}

\section{Discussion and outlook}
\label{sec:discussion}

In this work we looked into $U(1)$ covariant codes generated by Haar random unitaries, and proved that with overwhelming probability, the error rate of such codes as measured by Choi and worst-case purified distances scales as $O(n^{-1})$ and exactly saturates the lower bounds given in \cite{faist2019} at leading order in standard cases, and thus such random codes almost always have the best performance allowed in the presence of certain conservation laws.

\zw{We expect our results and methods to find connections or implications to several important problems in physics as hinted in the introduction, as it is often crucial to take charge or energy conservation laws into account in physical scenarios.  
We shall leave detailed explorations for future work, but below we make preliminary remarks on the possible relevance to a few specific directions and point out some references:
\begin{itemize}
    \item Black hole information problem and Hayden--Preskill thought experiment.  The Hayden--Preskill thought experiment  \cite{hayden2007} is a model of quantum information retrieval from black hole radiation based on scrambling dynamics modeled by e.g.~random circuits, which has stimulated many key developments in our understanding of black hole information problem and quantum gravity. The basic correspondences to our error correction setup is as follows: 
    The logical system corresponds to the quantum message thrown into the black hole, and the ancilla state corresponds to a pure state black hole (that has not evaporated) with fixed charge. Now the black hole dynamics subject to charge or energy conservation is modeled by our $U(1)$-covariant random code, the error of which  against erasing the remaining black hole characterizes how well the message can be recovered.  Note that our current analysis, especially the optimality arguments, are mostly for the regime of relatively small $t$ (size of erasure), so to understand the connections to Hayden--Preskill it could be important to further look into the large $t$ regime. 
    It is also discussed in \cite{liu2020} that the performance of Hayden--Preskill with charge conservation has possible conceptual connections to the weak gravity conjecture \cite{Arkani_Hamed_2007}---it is possible that our analysis leads to useful quantitative statements.
    We refer readers to e.g.~\cite{Yoshida:softmode,liu2020} for more background and discussions on Hayden--Preskill with conservation laws.
\item Scrambling and entanglement in many-body physics.  Random unitaries and circuits have also drawn considerable interest in condensed matter physics in recent years as ``solvable'' models of chaotic dynamics, leading to highly active fields like entanglement or operator spreading \cite{Nahum2,Nahum1} and measurement-induced entanglement transition \cite{PhysRevX.9.031009,PhysRevB.100.134306,PhysRevB.99.224307}.   Here conservation laws lead to diffusive transport of the conserved quantities, and as shown in e.g.~\cite{PhysRevX.8.031057}, the laws of scrambling (such as operator spreading) could be fundamentally different from the case without conservation laws.  On the other hand, the measurement-induced phase transition in random circuits can be understood from a quantum error correction perspective \cite{ChoiBaoQiAltman}, indicating further connections between quantum error correction and phases of matter. 
Given the importance of random unitaries, symmetries and quantum error correction in these areas, we would hope to further explore the implications of our rigorous understanding of the interplay between them.
\end{itemize}
As mentioned, symmetries and quantum error correction are key notions in holographic quantum gravity as well.  In particular, the famous conjecture about quantum gravity that it does not allow global symmetries is recently argued in the positive in AdS/CFT based on the quantum error correction formulation of AdS/CFT \cite{harlow2018,harlow2019}, and the argument can indeed be understood from the limitations of covariant codes \cite{faist2019}.  It could also be fruitful to further explore the applications of covariant code results and techniques there.  
}

\zw{
Another future direction is to consider more general symmetries, especially $SU(d)$, given its link to fault-tolerant quantum computing \cite{eastin2009,faist2019}. 
It seems possible to generalize our method to $SU(d)$, as the Hilbert space of $n$ qudits takes the ``direct-sum-product'' structure under the representation $U\to U^{\otimes n}$, so the partial decoupling theorem still applies. However, it remains to be calculated in detail how good the performance of the random code is. 
}

\zw{
With our work, we hope to stimulate further study into random unitaries and circuits with symmetries, which, as discussed earlier, exhibit many interesting distinctions from the no-symmetry case.  In this work we considered random global unitaries, but it would be important to further study random circuits since they can capture the ``complexity'' of the construction as well as the locality structure that is important in physical scenarios.     As a general program, it would be interesting to better understand the relations between Haar random unitaries, $t$-designs, and random local circuits (with various architectures), in the presence of $U(1)$ or other symmetries, through various kinds of physical properties and measures.  Here we take a first step by analyzing the error correction performance of random unitaries with charge conservation, and conjectured that it holds for low-depth random circuits.   As discussed there are difficulties in fully understanding the convergence of symmetric random circuits to designs, but it could already be interesting and useful to look into the behaviors of ``measures'' of scrambling and randomness, such as frame potentials \cite{Scott_2008,2016arXiv160908172Z,RobertsYoshida}, out-of-time-order correlators \cite{HQRY16,RobertsYoshida}, R\'enyi entanglement entropies \cite{HQRY16,LLZZ18,PhysRevLett.120.130502}, which are widely used in physics and quantum information.
}

\section*{Acknowledgements}
We thank Daniel Gottesman, Aram Harrow, Sirui Lu, Beni Yoshida, Sisi Zhou for useful discussions and feedback.   
LK is supported by the ARO grant Contract Number W911NF-12-0486. ZWL is supported by Perimeter Institute for Theoretical Physics.
Research at Perimeter Institute is supported in part by the Government of Canada through the Department of Innovation, Science and Economic Development Canada and by the Province of Ontario through the Ministry of Colleges and Universities.

\appendix
\section{Conditional min-entropy bounds for Choi error}
\label{app:entropy}
We first restate the partial decoupling theorem adapted to our structure of the space, which corresponds to $l_j=1$ and $r_j = \binom{n}{j}$ for all $0\le j \le n$ in \cite{wakakuwa2019}. We denote by $\Pi_j$ the projector into the subspace with Hamming weight $j$.
\begin{lemma}[Partial Decoupling]
Let $\mathcal{T}^{A\to E}$ be any channel mapping system $A$ to system $E$, and let $\Psi^{AR}$ be any joint state of system $A$ and $R$. We have
\[
    \mathbb{E}_{U \sim H_{\mathrm{x}}}\left[\left\|\mathcal{T}^{A \rightarrow E} \circ \mathcal{U}^{A}\left(\Psi^{A R}\right)-\mathcal{T}^{A \rightarrow E}\left(\Psi_{\avg}^{A R}\right)\right\|_{1}\right] \leq 2^{-\frac{1}{2} H_{\min }\left(A^{*} \mid R E\right)_{\Lambda(\Psi, \mathcal{T})}},
\]
where
\[
    \Psi_avg^{A R} = \E_{U\sim H_\times} U \Psi^{A R} U^\dagger.
\]
The state $\Lambda(\Psi, \mathcal{T})$ is defined as
\[
    \Lambda(\Psi, \mathcal{T}) = F(\Psi^{AR} \otimes \tau^{\bar A E})F^\dagger
\]
where $\tau^{\bar A E}$ is the Choi-Jamiolkowski state of $\mathcal{T}$ and the operator $F^{A\bar A \to A^*}$ is
\[
    F^{A \bar{A} \rightarrow A^{*}}:=\bigoplus_{j=0}^n \sqrt{\frac{2^n }{\binom{n}{j}}}\left(\Pi_{j}^{A} \otimes \Pi_{j}^{\bar{A}}\right).
\]
\end{lemma}

For simplicity we define the $2m$-qubit normalized state
\[
    |\phi_i^{(m)}\> = \binom{m}{i}^{-1/2} \sum_{v \in \{0,1\}^m,\,|v|=i}|v\>|v\>,
\]
which is the maximally entangled state between two copies of subspaces with Hamming weight $i$ of $m$ qubits. We also define $|\hat\phi^{(m)}\>$ as $m$ EPR pairs.

In our setting the state $\Psi^{AR}$ is the encoded state before applying the random unitary, which is $k$ EPR pairs appended by the fixed state $|\psi\>$,
\[
    \Psi^{AR} = |\Psi^{AR}\>\<\Psi^{AR}|,\quad |\Psi^{AR}\> =  |\hat \phi\>^{A_1 R} \otimes |\psi_\alpha\>^{A_2}.
\]
$A_1$ and $A_2$ refers to different parts of $A$, and have $k$ and $n-k$ qubits each. Then it is easy to see that
\[
    \Pi_j^A |\Psi\>^{AR} = \sqrt{\frac{\binom{k}{j-\alpha}}{2^k}}|\phi_{j-\alpha}^{(k)}\>^{A_1R}|\psi_\alpha\>^{A_2}.
\]

The channel $\mathcal T$ traces over $n-t$ qubits, so the corresponding Choi-Jamiolkowski is
\[
    \tau^{\bar A E} = \Tr_{n-t} |\hat\phi^{(n)}\>\<\hat\phi^{(n)}| = |\hat\phi^{(t)}\>\<\hat\phi^{(t)}|^{\bar A_1 E} \otimes \frac{I^{\bar A_2}}{2^{n-t}}.
\]
Let $\Pi^{(a)}_b$ be the subspace on $a$ qubits with Hamming weight $b$. We have
\begin{align}
    \Pi_j^{\bar A}\tau^{\bar A E}\Pi_{j'}^{\bar A} = & \frac{1}{2^{n-t}}\sum_i\Pi_j^{\bar A}\left[|\hat\phi^{(t)}\>\<\hat\phi^{(t)}|^{\bar A_1 E} (\Pi^{(n-t)}_i)^{\bar A_2}\right] \Pi_{j'}^{\bar A} \nonumber \\
    =&\frac{1}{2^{n}}\sum_i \sqrt{\binom{t}{j-i}\binom{t}{j'-i}} |\phi_{j-i}^{(t)}\>\<\phi_{j'-i}^{(t)}|^{\bar A_1 E} (\Pi^{(n-t)}_i)^{\bar A_2}, \label{eq:proj}
\end{align}
and therefore
\begin{align*}
    \Lambda(\Psi, \mathcal{T}) =& \frac{1}{2^{k}}\sum_{i,j,j'} \sqrt{\frac{\binom{t}{j-i}\binom{t}{j'-i}\binom{k}{j-\alpha}\binom{k}{j'-\alpha}}{\binom{n}{j}\binom{n}{j'}}} |\phi_{j-i}^{(t)}\>\<\phi_{j'-i}^{(t)}|^{\bar A_1 E} (\Pi^{(n-t)}_i)^{\bar A_2}|\phi_{j-\alpha}^{(k)}\>\<\phi_{j'-\alpha}^{(k)}|^{A_1R}|\psi_\alpha\>\<\psi_\alpha|^{A_2} \\
    =&  \sum_{i,j,j'} (\Pi^{(n-t)}_i)^{\bar A_2} \otimes |\gamma_{j,i}\>\<\gamma_{j',i}|^{\bar A_1 A_1A_2ER} \\
    =& \sum_i (\Pi^{(n-t)}_i)^{\bar A_2} \otimes |\Gamma_i\>\<\Gamma_i|^{\bar A_1 A_1A_2ER}
\end{align*}
where
\begin{align*}
    |\Gamma_i\> =& \sum_j |\gamma_{j,i}\>, \\
    |\gamma_{j,i}\>^{\bar A_1 A_1A_2ER} =& \sqrt{\frac{\binom{t}{j-i}\binom{k}{j-\alpha}}{2^k\binom{n}{j}}}|\phi_{j-i}^{(t)}\>^{\bar A_1 E}|\phi_{j-\alpha}^{(k)}\>^{A_1R}|\psi_\alpha\>^{A_2}.
\end{align*}

As mentioned in Section~\ref{subsec:min-entropy}, the min conditional entropy is defined as
\[
    H_{\min}(P|Q)_\rho = \sup_{\sigma\ge 0, \Tr\sigma=1}\sup\{\lambda\in \R| 2^{-\lambda}I^P\otimes \sigma^Q \ge \rho^{PQ}\},
\]
or equivalently, $H_{\min}(P|Q)_\rho=-\log s$ where $s$ is the optimum value of the following SDP
\begin{equation}
    s=\inf \Tr\sigma, \quad \text{s.t. }I^P\otimes \sigma^Q \ge \rho^{PQ},\,\sigma\ge0. \label{eq:primal}
\end{equation}
The corresponding dual SDP is
\begin{equation}
    t=\sup \<\rho^{PQ},y^{PQ}\>,\quad \text{s.t. }\Tr_P[y^{PQ}]\le I_Q,\,y\ge 0. \label{eq:dual}
\end{equation}
It is obvious that both primal and dual are strongly feasible, so $s=t$. We can use the following lemma to relate the min-entropy of $\Lambda(\Psi, \mathcal{T})$ to the min-entropy of each $|\Gamma_i\>$.
\begin{lemma}
\label{lem:sum}
Suppose the register $P$ in eqs.~\eqref{eq:primal}\eqref{eq:dual} can be divided into $P_1$ and $P_2$, and the state $\rho^{PQ}$ has the structure
\[
    \rho^{PQ} = \sum_{i=1}^m \Pi_i^{P_1} \otimes \rho_i^{P_2Q}
\]
where $\Pi_i$ are projectors into disjoint subspaces. Then $s$, the result of the SDP, satisfy 
\[
\frac{1}{m}\sum_i s_i \le  s \le \sum_i s_i
\]
where $s_i$ is the result for the SDP of $\rho_i$.
\end{lemma}

\begin{proof}
We prove the lemma by constructing feasible solutions of the primal and the dual. Let
\[
    \sigma=\sum_i \sigma_i
\]
where $\sigma_i$ is the optimal solution for the primal SDP of $\rho_i$. Then it is natural that $\Tr[\sigma]=\sum_i s_i$, and the condition holds because
\[
    I^{P_1P_2} \otimes \sigma \ge \sum_i \Pi_i^{P_1} \otimes I^{P_2} \otimes \sigma_i^Q \ge \sum_i \Pi_i^{P_1} \otimes \rho_i^{P_2Q} = \rho.
\]

For the dual SDP, let
\[
    y = \frac{1}{m}\sum_i \frac{\Pi_i^{P_1}}{\Tr[\Pi_i]} \otimes y_i^{P_2Q}
\]
where $y_i$ is the optimal solution for the dual SDP of $\rho_i$. Then
\[
    \<\rho,y\> = \frac{1}{m}\sum_i \Tr\left[\Pi_i\frac{\Pi_i}{\Tr[\Pi_i]}\right]\Tr[y_i\rho_i] = \frac{1}{m}\sum_i s_i,
\]
and
\[
    \Tr_P y = \frac{1}{m}\sum_i \Tr \frac{\Pi_i}{\Tr[\Pi_i]} \Tr_{P_2}[y_i] \le \frac{1}{m}\sum_i I^Q = I^Q.
\]
\end{proof}
Note that the state $|\Gamma_i\>$ is a pure state, so its min entropy can be calculated using Eq.~\eqref{eq:pure-min}. Therefore we have
\begin{align*}
    H_{\min}(A^*|RE)_{\Gamma_i} = -2\log\left(\sum_j \frac{1}{\sqrt{2^k\binom{n}{j}}}\binom{t}{j-i}\binom{k}{j-\alpha}\right),
\end{align*}
and the value for the corresponding SDP is
\[
  \left(\sum_j \frac{1}{\sqrt{2^k\binom{n}{j}}}\binom{t}{j-i}\binom{k}{j-\alpha}\right)^2. 
\]

Note that $j$ and $i$ should satisfy
\[
    0 \le j-i \le t, \quad 0 \le j-\alpha \le k,
\]
so $\alpha-t \le i \le \alpha + k$, and there are at most $k+t+1$ possible values for $i$. By Lemma~\ref{lem:sum} we have
\begin{equation}
     -\log\kappa \le H_{\min}(A^*|RE)_{\Lambda} \le  -\log\frac{\kappa}{k+t+1}. \label{eq:entropy-bound}
\end{equation}
where
\[
    \kappa = \sum_i\left(\sum_j \frac{1}{\sqrt{2^k\binom{n}{j}}}\binom{t}{j-i}\binom{k}{j-\alpha}\right)^2.
\]
Note that
\[
    \frac{1}{\sqrt{2^k\binom{n}{j}}}\binom{t}{j-i}\binom{k}{j-\alpha} \le 2^{-k/2}\binom{t}{t/2}\binom{k}{k/2}\frac{1}{\sqrt{\min\{\binom{n}{\alpha},\binom{n}{\alpha+k}\}}},
\]
so from Eq.~\eqref{eq:entropy-bound} we have
\[
    H_{\min}(A^*|RE)_{\Lambda} \ge n \min\left\{H\left(\frac{\alpha}{n}\right),H\left(\frac{\alpha+k}{n}\right)\right\}-2t-k+O(\log n)
\]
for general values of $t$ and $k$ as long as $\alpha$ is linear in $n$. Here $H(\cdot)$ is the binomial entropy function
\[
    H(x)= -x\log x -(1-x)\log(1-x), \quad 0 \le x \le 1.
\]
If $t^2k^2=O(n)$, this implies $H_{\min}(A^*|RE)_{\Lambda} = \Omega(n)$.

\linghang{
When $\alpha$, $k$ and $t$ are all $O(1)$ and does not depend on $n$, the bound in Eq.~\eqref{eq:entropy-bound} imples that
\[
     H_{\min}(A^*|RE)_{\Lambda} \ge \alpha\log n + O(1).
\]
}

\section{Average state and Choi error}
\label{app:average-state}
Following the previous definitions, $\Phi_\text{avg}$ is a joint state on $n$-qubit register $A$ and $k$-qubit register $R$. From Eq.~\eqref{eq:avg-state},
\begin{align*}
    \Phi_\text{avg}^{RA} =& \E_{U\sim H_\times} U(|\hat\phi^{(k)}\>\<\hat\phi^{(k)}|^{A_1R}\otimes |\psi_\alpha\>\<\psi_\alpha|^{A_2})U^\dagger \\
    =& \sum_{j=0}^k \Tr_A[\Pi_j^A(|\hat\phi^{(k)}\>\<\hat\phi^{(k)}|^{A_1R}\otimes |\psi_\alpha\>\<\psi_\alpha|^{A_2})\Pi_j^A] \otimes \frac{\Pi_j^A}{\binom{n}{j}} \\
    =& 2^{-k}\sum_{j=\alpha}^{k+\alpha} \Pi_{j-\alpha}^R \otimes \frac{\Pi_j^A}{\binom{n}{j}},
\end{align*}
and
\[
    \Tr_{n-t}\Phi_\text{avg}^{RA}=2^{-k}\sum_{j=\alpha}^{k+\alpha} \sum_{i=0}^{t} \Pi_{j-\alpha}^R \otimes \Pi_i^E\frac{\binom{n-t}{j-i}}{\binom{n}{j}} = 2^{-k}\sum_{j=0}^{k} \sum_{i=0}^{t} \Pi_{j}^R \otimes \Pi_i^E\frac{\binom{n-t}{j+\alpha-i}}{\binom{n}{j+\alpha}},
\]
where $E$ refers to the $t$-qubit register that the complementary channel maps to. To get an upper bound for the second term of Eq.~\eqref{eq:bound}, we can replace the minimization over $\zeta$ by an arbitrary fixed $\zeta_0$, which we choose to be the marginal state
\[
    \zeta_0^E=\Tr_R \Tr_{n-t}\Phi_\text{avg}^{RA} = 2^{-k}\sum_{i=0}^{t}\Pi_i^E \sum_{j=\alpha}^{k+\alpha}  \frac{\binom{k}{j-\alpha}\binom{n-t}{j-i}}{\binom{n}{j}}.
\]
We define
\begin{equation}
    \quad \beta_i = 2^{-k}\sum_{j=\alpha}^{k+\alpha}  \frac{\binom{k}{j-\alpha}\binom{n-t}{j-i}}{\binom{n}{j}} = 2^{-k}\sum_{j=0}^{k}  \frac{\binom{k}{j}\binom{n-t}{j+\alpha-i}}{\binom{n}{j+\alpha}},\label{eq:beta}
\end{equation}
so that
\[
    \zeta_0^E = \sum_{i=0}^{t}\beta_i\Pi_i^E
\]

Note that all the states are diagonal, the fidelity is given by
\begin{align}
    F\left(\Tr_{n-t}\Phi_\text{avg},\frac{I}{2^k}\otimes \zeta_0\right) =& \sum_{j=0}^{k} \sum_{i=0}^{t} \Tr[\Pi_j^R \otimes \Pi_i^A]\sqrt{2^{-k}\beta_i \times 2^{-k}\frac{\binom{n-t}{j+\alpha-i}}{\binom{n}{j+\alpha}}} \nonumber \\ =&
    2^{-k}\sum_{j=0}^{k} \sum_{i=0}^{t} \binom{k}{j}\binom{t}{i}\sqrt{\beta_i \frac{\binom{n-t}{j+\alpha-i}}{\binom{n}{j+\alpha}}} \label{eq:choi-fidelity}
\end{align}

For any nonnegative number $n$ and real number $x$, we define
\[
    x^{\underline n} = x(x-1)\ldots (x-n+1),\quad x^{\overline n} = x(x+1)\ldots (x+n-1).
\]
They are related by
\[
    x^{\ol n} = (x+n-1)^{\ul n}, \quad x^{\ul n} = (x-n+1)^{\ol n}
\]
and
\[
    x^{\ul n} = (-x)^{\ol n} (-1)^n, \quad x^{\ol n} = (-x)^{\ul n} (-1)^n.
\]

It could be verified that the following binomial theorems hold (by induction on $n$)
\[
    (x+y)^{\underline n} = \sum_{k=0}^n \binom{n}{k} x^{\ul k} y^{\ul{n-k}},\quad (x+y)^{\ol n} = \sum_{k=0}^n \binom{n}{k} x^{\ol k} y^{\ol{n-k}}.
\]
Then from the definition of $\beta_i$ in Eq.~\eqref{eq:beta} we have
\begin{align*}
    \beta_i =& 2^{-k}\sum_{j=0}^{k}  \frac{\binom{k}{j}\binom{n-t}{j+\alpha-i}}{\binom{n}{j+\alpha}} \\
    =& \frac{1}{n^{\ul t}2^k} \sum_{j=0}^k \binom{k}{j} {(j+\alpha)^{\ul i}(n-\alpha-j)^{\ul {t-i}}} \\
    =& \frac{1}{n^{\ul t}2^k} \sum_{j,x} \binom{k}{j} \binom{i}{x}j^{\ul x}\alpha^{\ul {i-x}}(-1)^{t-i}(j-n+\alpha)^{\ol {t-i}} \\
    =& \frac{1}{n^{\ul t}2^k} \sum_{j,x} \binom{k}{j} \binom{i}{x}j^{\ul x}\alpha^{\ul {i-x}}(-1)^{t-i}(j-n+\alpha+t-i-1)^{\ul {t-i}} \\
    =& \frac{1}{n^{\ul t}2^k} \sum_{j,x,y} \binom{k}{j} \binom{i}{x}j^{\ul x}\alpha^{\ul {i-x}}(-1)^{t-i}\binom{t-i}{y}(j-x)^{\ul y}(x-n+\alpha+t-i-1)^{\ul {t-i-y}} \\
    =& \frac{1}{n^{\ul t}2^k} \sum_{j,x,y}\binom{k}{j}j^{\ul{x+y}}  \binom{i}{x}\alpha^{\ul {i-x}}(-1)^{y}\binom{t-i}{y}(-x+n-\alpha-t+i+1)^{\ol {t-i-y}} \\
    =& \frac{1}{n^{\ul t}}  \sum_{x,y} 2^{-(x+y)}(-1)^{y}k^{\ul{x+y}} \binom{i}{x}\binom{t-i}{y}\alpha^{\ul {i-x}}(n-\alpha-x-y)^{\ul {t-i-y}}. \\
\end{align*}
Suppose that $kt=o( n)$, we can see that the term with $x=x_0$ and $y=y_0$ is of order $(kt/n)^{x_0+y_0}$ times the term with $x=y=0$. Then by keeping the terms $x+y \le 2$ we have the series expansion
\[
    \beta_i = \frac{\alpha^{\ul i}(n-\alpha)^{\ul{t-i}}}{n^{\ul t}}\left(1+\frac{i k-a k t}{\left(2 a-2 a^2\right) n}+\frac{\xi_1}{8 (a-1)^2 a^2 n^2}+O\left(\frac{k^3t^3}{n^3}\right)\right),
\]
where $a=\alpha/n$ and
\[
    \xi_1 = k \left(a^2 (k-1) (t-1) t+i^2 (-4 a+k+3)+i (2 a (2 a (t-1)-(k+1) t+k+3)-k-3)\right).
\]
Similarly
\[
    \frac{\binom{n-t}{j+\alpha-i}}{\binom{n}{j+\alpha}} = \frac{\alpha^{\ul i}(n-\alpha)^{\ul{t-i}}}{n^{\ul t}}\left(1+\frac{j (a t-i)}{(a-1) a n}+\frac{\xi_2}{2 (a-1)^2 a^2 n^2}+O\left(\frac{k^3t^3}{n^3}\right)\right)
\]
with
\[
    \xi_2 = j \left(a^2 (j-1) (t-1) t+i^2 (-2 a+j+1)+i (2 a (a (t-1)+j (-t)+j+1)-j-1)\right).
\]

Therefore,
\[
    \sqrt{\beta_i \frac{\binom{n-t}{j+\alpha-i}}{\binom{n}{j+\alpha}}} = \frac{\alpha^{\ul i}(n-\alpha)^{\ul{t-i}}}{n^{\ul t}}\left(1+\frac{(2 j+k) (a t-i)}{4 (a-1) a n}+\frac{\xi_3}{32 (a-1)^2 a^2 n^2}+O\left(\frac{k^3t^3}{n^3}\right)\right),
\]
where
\begin{align*}
    \xi_3 =& a^2 t \left(4 j^2 (t-2)+4 j ((k-2) t+2)+k (k (t-2)-2 t+2)\right) \\
    &+i^2 \left(4 j (-4 a+k+2)+k (-8 a+k+6)+4 j^2\right)\\
    &+2 i \left(-4 j^2 (a (t-2)+1)+j (4 a (2 a (t-1)-k t+2)-4)\right) \\
    &-2ik\left( (a (-4 a (t-1)+(k+2) t-2 (k+3))+k+3)\right).
\end{align*}
Then we can multiply this by $\binom{t}{i}$ and sum over $i$ to have the fidelity
\begin{equation}
    \sum_i \binom{t}{i}\sqrt{\beta_i \frac{\binom{n-t}{j+\alpha-i}}{\binom{n}{j+\alpha}}} = 1-\frac{t(k-2j)^2}{32a(1-a)n^2}+O\left(\frac{k^3t^3}{n^3}\right). \label{eq:fidelity}
\end{equation}
Now we plug this into Eq.~\eqref{eq:choi-fidelity} and have
\[
    F\left(\Tr_{n-t}\Phi_\text{avg},\frac{I}{2^k}\otimes \zeta_0\right) = 1-\frac{tk}{32a(1-a)n^2}+O\left(\frac{k^3t^3}{n^3}\right),
\]
and the corresponding purified distance is
\[
    P\left(\Tr_{n-t}\Phi_\text{avg},\frac{I}{2^k}\otimes \zeta_0\right) = \frac{\sqrt{tk}}{4n\sqrt{a(1-a)}}\left(1+O\left(\frac{t^2k^2}{n}\right)\right).
\]
\linghang{
Another interesring case to consider is $\alpha=O(1)$, which could be directly evaluated from Eq.~\eqref{eq:choi-fidelity} for small values of $k$ and $t$. For example, when $k=t=1$, we have
\[
    F\left(\Tr_{n-t}\Phi_\text{avg},\frac{I}{2^k}\otimes \zeta_0\right) = 1+\frac{-\sqrt{2}-2\sqrt{2}\alpha+\sqrt{\alpha(2\alpha+1)}+\sqrt{(\alpha+1)(2\alpha+1)}}{2\sqrt{2}n}+O(n^{-2}),
\]
and when $k=t=2$,
\[
    F\left(\Tr_{n-t}\Phi_\text{avg},\frac{I}{2^k}\otimes \zeta_0\right) = 1+\frac{-2-2 \alpha +\sqrt{\alpha  (\alpha +1)}+\sqrt{(\alpha +1) (\alpha +2)}}{2 n}+1+O(n^{-2}).
\]
In both cases the purified distance is $O(n^{-1/2})$.
}
\section{Conditional min-entropy bounds for worst-case error}
\label{app:worst-entropy}
We are interested in the conditional entropy $H_{\min }\left(A^{*} \mid R E\right)_{\Lambda(\Psi, \mathcal{T})}$ with the initial states $\Psi$ being $|x\>$, $|\mu^\pm_{x,x'}\>$ and $|\nu^\pm_{x,x'}\>$, where
\[
    |\mu^\pm_{x,x'}\> = \frac{1}{\sqrt 2}(|x\> \pm |x'\>), \quad|\nu^\pm_{x,x'}\> = \frac{1}{\sqrt 2}(|x\> \pm i|x'\>).
\]
In other words, the state $|\Psi^{AR}\>$ is one of the above states appended by $|\psi_\alpha\>$, a state with Hamming weight $\alpha$. The reference system $R$ is now trivial, in contrast to the $k$ qubits in Appendix~\ref{app:entropy}. Here the channel $\mathcal{T}$ is the erasure channel over $n-t$ qubits.  Using Eq.~\eqref{eq:proj}, we have
\begin{align}
    \Lambda(\Psi, \mathcal{T}) =& \frac{1}{\binom{n}{|x|+\alpha}} |x\>\<x|^{A_1} \otimes |\psi_\alpha\>\<\psi_\alpha|^{A_2} \otimes \sum_i \binom{t}{|x|+\alpha-i}|\phi_{|x|+\alpha-i}^{(t)}\>\<\phi_{|x|+\alpha-i}^{(t)}|^{\bar A_1 E} (\Pi_i^{(n-t)})^{\bar A_2} \nonumber\\
    =& \sum_i (\Pi_i^{(n-t)})^{\bar A_2} \otimes |\gamma_{x,i}\>\<\gamma_{x,i}|, \label{eq:lambda}
\end{align}
where
\[
    |\gamma_{x,i}\> = \sqrt{\frac{\binom{t}{|x|+\alpha-i}}{\binom{n}{|x|+\alpha}}}|x\>^{A_1} \otimes |\psi_\alpha\>^{A_2} \otimes |\phi_{|x|+\alpha-i}^{(t)}\>^{\bar A_1 E}.
\]
Now using Lemma~\ref{lem:sum} and Eq.~\eqref{eq:pure-min}, we have
\[
   -\log\left[ \frac{\sum_i\binom{t}{|x|+\alpha-i}^2}{\binom{n}{|x|+\alpha}}\right]  \le  H_{\min }^x \le -\log\left[\frac{\sum_i\binom{t}{|x|+\alpha-i}^2}{\binom{n}{|x|+\alpha}(t+1)}\right], 
\]
where $H_{\min }^x$ stands for $H_{\min }\left(A^{*} \mid R E\right)_{\Lambda(\Psi, \mathcal{T})}$ when the initial state is $|x\>\<x|$. This could be further simplified to
\[
   -\log\left[ \frac{\binom{2t}{t}}{\binom{n}{|x|+\alpha}}\right]  \le  H_{\min }^x \le -\log\left[\frac{\binom{2t}{t}}{\binom{n}{|x|+\alpha}(t+1)}\right]. 
\]
When the initial state is $|\mu^\pm_{x,x'}\>$ and $|\nu^\pm_{x,x'}\>$, the state in Eq.~\eqref{eq:lambda} will have the same form, with $|\gamma_{x,i}\>$ replaced by $\frac{1}{\sqrt{2}}(|\gamma_{x,i}\>\pm |\gamma_{x',i}\>)$ and $\frac{1}{\sqrt{2}}(|\gamma_{x,i}\>\pm i|\gamma_{x',i}\>)$ correspondingly. Then
\[
   -\log\left[\frac{\chi}{2} \right]  \le  H_{\min }^{x,x'} \le -\log\left[\frac{\chi}{2(t+1)} \right],
\]
where $H_{\min }^{x,x'}$ stands for $H_{\min }\left(A^{*} \mid R E\right)_{\Lambda(\Psi, \mathcal{T})}$ when the initial state is one of $|\mu^\pm_{x,x'}\>$ or $|\nu^\pm_{x,x'}\>$, and
\[
    \chi = \sum_i\left(\frac{\binom{t}{|x|+\alpha-i}}{\sqrt{\binom{n}{|x|+\alpha}}}+\frac{\binom{t}{|x'|+\alpha-i}}{\sqrt{\binom{n}{|x'|+\alpha}}}\right)^2 = \frac{\binom{2t}{t}}{\binom{n}{|x|-\alpha}}+\frac{\binom{2t}{t}}{\binom{n}{|x'|-\alpha}}+\frac{2\binom{2t}{t+|x|-|x'|}}{\sqrt{\binom{n}{|x|-\alpha}\binom{n}{|x'|-\alpha}}}.
\]

Suppose that in the large $n$ limit $\frac{\alpha}{n}$ and $\frac{\alpha+k}{n}$ are both constants between 0 and 1, we have
\[
    H_{\min }^x = n H\left(\frac{|x|+\alpha}{n}\right) -2t +O(\log n)\ge n \min\left\{H\left(\frac{\alpha}{n}\right),H\left(\frac{\alpha+k}{n}\right)\right\} - 2t + O(\log n),
\]
and
\begin{align*}
    H_{\min }^{x,x'} \ge & n \min\left\{H\left(\frac{|x|+\alpha}{n}\right),H\left(\frac{|x'|+\alpha}{n}\right)\right\} - 2t + O(\log n) \\ 
    \ge & n \min\left\{H\left(\frac{\alpha}{n}\right),H\left(\frac{\alpha+k}{n}\right)\right\} - 2t + O(\log n).
\end{align*}
\section{Average state and worst-case error}
\label{app:worst-avg}
It is easy to see that
\[
    \E_{U\sim H_\times} U(|x\>\<x'| \otimes |\psi\>\<\psi|) U^\dagger =
    \begin{cases}
        0, & x \not= x' \\
        \Pi^{(n)}_{|x|+\alpha}/\binom{n}{|x|+\alpha}, &x = x'.
    \end{cases}
\]
In the case of $x=x'$, we take trace over $n-t$ qubits and have
\[
    \rho^{x,x}_\text{avg} = \sum_{i=0}^t \Pi_i^{(t)}\frac{\binom{n-t}{|x|+\alpha-i}}{\binom{n}{|x|+\alpha}},
\]
and we wish to show that this is close to some fixed state $\zeta$ independent of $x$. We propose that $\zeta$ is $\rho^{x,x}_\text{avg}$ averaged over $x$,
\[
    \zeta = \frac{1}{2^k}\sum_{j=0}^k\Pi_i^{(t)}\frac{\binom{k}{j}\binom{n-t}{j+\alpha-i}}{\binom{n}{j+\alpha}}=\sum_i \beta_i \Pi_i^{(t)},
\]
where $\beta_i$ is the quantity previously defined in Eq.~\eqref{eq:beta} from Appendix~\ref{app:average-state}. For any $x$, the fidelity is given by
\[
    F(\rho^{x,x}_\text{avg},\zeta) = \sum_i \binom{t}{i} \sqrt{\beta_i \frac{\binom{n-t}{|x|+\alpha-i}}{\binom{n}{|x|+\alpha}}}.
\]
which is exactly the result in Eq.~\eqref{eq:fidelity} with $j=|x|$. Then we have
\[
    \max_x P(\rho^{x,x}_\text{avg},\zeta) = \frac{k\sqrt{t}}{4n\sqrt{a(1-a)}}\left(1+\left(\frac{kt^2}{n}\right)\right).
\]

\bibliography{covariance}
\bibliographystyle{unsrt}

\end{document}